\documentclass[a4paper,USenglish,cleveref, autoref, thm-restate]{lipics-v2019}

\title{Space Hardness of Solving Structured Linear Systems} 

\author{Xuangui Huang}{Northeastern University, U.S.}{stslxg@ccs.neu.edu}{}{}

\authorrunning{X. Huang} 

\Copyright{Xuangui Huang} 



\category{} 

\relatedversion{} 

\supplement{}

\funding{Supported by NSF CCF award 1813930.}

\acknowledgements{The author thanks Emanuele Viola for helpful discussions.}

\nolinenumbers 

\EventEditors{John Q. Open and Joan R. Access}
\EventNoEds{2}
\EventLongTitle{42nd Conference on Very Important Topics (CVIT 2016)}
\EventShortTitle{CVIT 2016}
\EventAcronym{CVIT}
\EventYear{2016}
\EventDate{December 24--27, 2016}
\EventLocation{Little Whinging, United Kingdom}
\EventLogo{}
\SeriesVolume{42}
\ArticleNo{23}

\usepackage[ruled,linesnumbered,commentsnumbered]{algorithm2e}
\newtheorem{fact}[theorem]{Fact}
\renewcommand{\epsilon}{\varepsilon}

\renewcommand{\tilde}{\widetilde}

\newcommand{\AC}{\mathsf{AC}}

\newcommand{\TC}{\mathsf{TC}}
\newcommand{\ZZ}{\mathbb{Z}}
\newcommand{\RR}{\mathbb{R}}
\newcommand{\mtt}[1]{\mathtt{#1}}

\def\*#1{\mathbf{#1}}
\def\+#1{\mathcal{#1}}
\def\-#1{\mathbb{#1}}

\newcommand{\MCG}{\+{MC}_2\mathtt{Gadget}}
\newcommand{\NNV}{\mtt{NumGadget}}
\newcommand{\CB}{\mtt{CountBit}}
\newcommand{\Len}{\mtt{Len}}
\newcommand{\SNNV}{\mtt{SumNumG}}
\newcommand{\PreS}{\mtt{PrefixSum}}

\newcommand{\mn}[1]{}
\newcommand{\xn}[1]{}

\allowdisplaybreaks

\begin{document}

\global\long\def\poly{\mathrm{{poly}}}%
\global\long\def\polylog{\mathrm{{polylog}}}%

\global\long\def\zo{\mathrm{\{0,1\}}}%

\global\long\def\mo{\mathrm{\{-1,1\}}}%

\global\long\def\e{\mathrm{\epsilon}}%

\global\long\def\eps{\mathrm{\epsilon}}%

\global\long\def\E{\mathrm{\mathbb{E}}}%

\global\long\def\F{\mathrm{\mathbb{F}}}%

\global\long\def\P{\mathrm{\mathbb{P}}}%

\maketitle

\begin{abstract}
We show that if the probabilistic logarithmic-space solver or the deterministic nearly logarithmic-space solver for undirected Laplacian matrices can be extended to solve slightly larger subclasses of linear systems, then they can be use to solve all linear systems with similar space complexity.
Previously Kyng and Zhang \cite{KyngZ17} proved similar results in the time complexity setting using reductions between approximate solvers.
We prove that their reductions can be implemented using constant-depth, polynomial-size threshold circuits.
\end{abstract}

\section{Introduction}
One of the oldest mathematical problems is to solve a linear system, that is to find a solution $\*x$ satisfying $\*A\*x = \*b$ given an $n \times n$  matrix $\*A$ and a $n$-dimensional vector $\*b$ as input.
In the RealRAM model this can be done in $O(n^{\omega})$ time, where $\omega \leq 2.3728$ \cite{Gall14a} is the matrix multiplication constant.
Much faster algorithms exist for approximately solving linear systems when $\*A$ is the Laplacian of undirected graphs.
Indeed recent breakthroughs showed that it can be done in nearly linear time \cite{SpielmanT14,CohenKMPPRX14}.
Kyng and Zhang \cite{KyngZ17} further showed that if such solvers can be extended to nearly linear time solvers for some classes slightly larger than undirected Laplacians, we can also solve general linear systems in nearly linear time.

In this paper we are interested in the space complexity of this problem.
For general linear systems Ta-shma gave a quantum algorithm using logarithmic space \cite{Ta-Shma13}.
For undirected Laplacian, Doron et al. showed that it has a probabilistic logarithmic-space algorithm \cite{DoronGT17} and hence a deterministic $O(\log^{3/2} n)$-space algorithm by a well-known space-efficient derandomization result \cite{SaksZ95}.
This was improved later to $\tilde{O}(\log n)$ by Murtagh et al \cite{MurtaghRSV17}.

\subsection{Our results}
We prove a space hardness version of Kyng and Zhang's results \cite{KyngZ17}, showing space hardness of approximate linear system solvers for some classes slightly larger than undirected Laplacians, namely multi-commodity Laplacians, 2D Truss Stiffness Matrices, and Total Variation Matrices.

\begin{theorem}\label{thm:hard}
Suppose that for multi-commodity Laplacians, 2-D Truss Stiffness Matrices, or Total Variation Matrices, the linear system $\*A\*x = \*b$ can be approximately solved in (nearly) logarithmic space with logarithmic dependence on condition number $\kappa$ and accuracy $\e^{-1}$ (even if it only works in expectation or with high probability), then any linear systems with polynomially bounded integers and condition number can be solved in (nearly) logarithmic space with high accuracy (in expectation or with high probability, respectively).
\end{theorem}

This shows that if the probabilistic logspace solver for undirected Laplacian in \cite{DoronGT17}, or the deterministic $\tilde{O}(\log n)$-space one in \cite{MurtaghRSV17}, can be extended to solve any of these three slightly larger subclasses of linear systems, we would have a surprising result that all linear systems can be approximately solved in probabilistic logspace or in deterministic $\tilde{O}(\log n)$-space.
Pessimistically speaking the theorem means that it is very hard to get space efficient algorithms for these subclasses, as it is as difficult as solving all linear systems in a space efficient way.
On the bright side, we actually prove that any progress on solving these subclasses using less space will immediately translate into similar progress for solving all linear system using less space.

Kyng and Zhang \cite{KyngZ17} proved their results via reductions from approximate solvers of general linear systems to those of three subclasses.
In this paper we prove Theorem \ref{thm:hard} by proving that their reductions can be carried out in a space efficient manner.
Indeed we prove a much stronger result that their reductions can be implemented in $\TC^0$ circuits, 
which are constant-depth polynomial-size unbounded-fan in circuits with $\mathsf{MAJORITY}$ and $\neg$ gates.
It shows that these reductions are actually highly parallelizable.

We denote $\+G$ as the class of all matrices with integer valued entries.
In the context of solving linear systems, an all-zero row or column can be trivially handled, so we can assume without loss of generality that matrices in $\+G$ has at least one non-zero entry in every row and column.
For 2-commodity matrices $\+{MC}_2$, we have two set of variables $X$ and $Y$ of the same size, and the equations are scalings of $x_i - x_j = 0$, $y_i - y_j = 0$, and $x_i - y_i - (x_j - y_j) = 0$, where $x_i, x_j \in X$ and $y_i, y_j \in Y$.
This generalizes undirected Laplacians, as the incidence matrices of undirected Lapacians only have equations of the form $x_i - x_j = 0$ for $x_i, x_j \in X$.

Our main technical result proves that the reduction from $\+G$ to $\+{MC}_2$ in \cite{KyngZ17} is $\TC^0$-computable.
\begin{theorem}\label{thm:main}
There is a $\TC^0$-reduction from approximately solving $\+G$ to approximately solving $\+{MC}_2$.
\end{theorem}

In \cite{KyngZ17} it is shown that the matrix produced by this reduction is also a 2D Truss Stiffness Matrix as well as a Total Variation Matrix, therefore Theorem \ref{thm:main} also works for these classes.
Also note that this reduction is a Karp-style reduction, i.e. it requires only one linear system solve and uses the solution in a black-box way.
That is why Theorem \ref{thm:hard} still applies if the solver only works in expectation or with high probability.

We also show $\TC^0$-computability of the reductions in \cite{KyngZ17} to some more restrictive subclasses of $\+{MC}_2$, including $\+{MC}^>_2$, the exact class we have to solve when we use Interior Point Methods for 2-commodity problems, as explained in the Kyng and Zhang's full paper \cite{KyngZ17full}.
They also showed that the promise problem of deciding if a vector is in the image of a matrix or $\e$-far from the image can be directly reduced to approximately solving linear systems.
Combining with the above results, this shows that the promise problem can be reduced to approximately solving the above-mentioned subclasses in $\TC^0$.

\section{Simplified reductions in $\TC^0$: the easy parts}
Throughout this paper we use the sign-magnitude representation to encode a $w$-bit signed integer $z \in [-2^{w+1}+1, 2^{w+1}-1]$ into a sign bit, $0$ for positive and $1$ for negative, and $w$ bits for $|z|$ in binary. Note that in this method, $0$ has two representations and we accept both.

For simplicity we first prove the special case for the reduction from exact solver of $\+G$ to exact solver of $\+{MC}_2$.
\begin{theorem}
Given an $m \times n$ $w$-bit matrix $\*A \in \+G$ and a vector $\*b$ as input, we can compute in $\TC^0$ an $O(nmw \log n) \times O(nmw \log n)$ $O(w \log n)$-bit matrix $\*A' \in \+{MC}_2$ and a vector $\*b'$ such that:
\begin{itemize}
	\item $\*A\*x = \*b$ has a solution if and only if $\*A'\*x' = \*b'$ has a solution;
	\item if $\*A'\*x' = \*b'$ has a solution $\*x'$, then we can compute $\*x$ in $\TC^0$ from $\*x'$ so that $\*A\*x = \*b$.
\end{itemize}
\end{theorem}

As in \cite{KyngZ17}, this reduction is split into several steps using the following classes of matrices.
\begin{definition}[\cite{KyngZ17}]
\begin{itemize}
\item Let $\+G_z \subset \+G$ denote the class of matrices with integer valued entries such that every row has zero row sum;
\item Let $\+G_{z,2} \subset \+G_z$ denote the class of matrices with integer valued entries such that every row has zero row sum, and for each row the sum of positive coefficients is a power of $2$.
\end{itemize}
\end{definition}

\begin{lemma}\label{lem:reduc}
There are $\TC^0$-reductions for exact solvers of the following classes:
\begin{enumerate}[(i)]
\item\label{red:1} from $\+G$ to $\+G_z$;
\item\label{red:2} from $\+G_z$ to $\+G_{z,2}$;
\item\label{red:3} from $\+G_{z,2}$ to $\+{MC}_2$.
\end{enumerate}
\end{lemma}

Lemma \ref{lem:reduc} (\ref{red:3}) is the main reduction in this paper (same as in \cite{KyngZ17}), which will be proved in the next section. In the remaining of this section we prove Lemma \ref{lem:reduc} (\ref{red:1}) and (\ref{red:2}).

\paragraph*{From $\+G$ to $\+G_z$}
\begin{proof}[Proof sketch]
Given a matrix $\*A \in \+G$, we can define a matrix $\*A'$ with one more column by $\*A' = \begin{pmatrix}
\*A & -\*A\*1
\end{pmatrix}$.
Obviously $\*A' \in \+G_z$, and $\*A \*x = \*b$ has a solution if and only if $\*A' \*x' = \*b$ has a solution.
We can also recover $\*x \in \mathbb{R}^n$ from $\*x' \in \mathbb{R}^{n+1}$ by taking the first $n$ rows of $\*x'$ and minus each of them by $\*x'_{n+1}$.
The following results about additions imply that $\*A'$ can be calculated in $\TC^0$, and we can recover $\*x'$ from $\*x$ in $\AC^0$ (for simplicity we ignore the precision problem here).
\end{proof}
\begin{fact}\label{fat:add}
\begin{itemize}
\item Addition of $2$ $w$-bit numbers has $\AC^0$ circuit of size $\poly(w)$ (c.f. \cite{CloteK02});\footnote{$\AC^0$ circuits are constant-depth polynomial-size unbounded-fan in circuits with $\wedge$, $\vee$, and $\neg$ gates.}
\item Addition of $n$ $w$-bit numbers has $\TC^0$ circuit of size $\poly(n,w)$ \cite{ImmermanL89}.
\end{itemize}
\end{fact}

\paragraph*{From $\+G_z$ to $\+G_{z,2}$}
\begin{proof}[Proof sketch]
Given an $m \times n$ $w$-bit signed-integer matrix $\*A' \in \+G_z$, we just need to add two more columns to make the sum of positive (and negative) entries in each row to the closet power of $2$. This can be done in $\TC^0$ in the following way. For each row $1 \leq i \leq m$, we calculate the sum of positive entries $s_i$ by checking the sign bit then do the iterated addition in $\TC^0$ by Fact \ref{fat:add}.
We then take $s = \max s_i$, which can be computed in $\AC^0$ given $s_i$'s.
$s$ has at most $O(w \log n)$ bits so given $s$ by searching brute-forcely in $\AC^0$ we can find the minimum $k$ s.t. $2^{k} \geq s$.
Therefore by Fact \ref{fat:add} we can calculate $a_i = 2^{k} - s_i$ in $\AC^0$ given $s_i$.
Then for each row $i$ we add $a_i$ and $-a_i$ to the last two columns of $\*A'$ to get $\*A''$. Additionally we need to add a new row to $\*A''$ (and to $\*b''$ accordingly) to zero out the last two variables we just added: set the last two entries into $1$ and $-1$, and set all other entries $0$, and also add a $0$ entry to $\*b'$ to get $\*b''$.

Obviously we have $\*A'' \in \+G_{z,2}$, and $\*A'\*x' = \*b'$ has a solution iff $\*A''\*x'' = \*b''$ has a solution. We can easily recover $\*x'$ from $\*x''$ by taking the first $n$ rows.
\end{proof}

For the next section we need the following definition.

\begin{definition}
We say a matrix $\*A \in \+G_{z}$ is \emph{$w$-bounded} if for each row the sum of positive coefficients is at most $2^w$. Note that in such matrix every entry is a $w$-bit signed integer.
\end{definition}

Note that the reduction from $\+G$ to $\+G_z$ reduces an $m \times n$ $w$-bit matrix $\*A$ into an $m \times (n+1)$ $O(w \log n)$-bounded matrix $\*A'$, then the reduction from $\+G_z$ to $\+G_{z,2}$ reduces $\*A'$ into an $(m+1) \times (n+3)$ $O(w \log n)$-bounded matrix $\*A''$.

\section{Simplified main reduction in $\TC^0$}
The main reduction in \cite{KyngZ17} uses the pair-and-replace scheme to transform a linear system in $\+G_{z,2}$ to a $2$-commodity equation system. The main idea is to use $\+{MC}_2\mathtt{Gadget}$s consisting of $\+{MC}_2$ equations to replace pairs of variables in the original linear system round-by-round according to the bit representation of their coefficients. 
A simplified version of the gadget, implicitly given in \cite{KyngZ17}, is as follows.

\begin{definition}[Simplified $\+{MC}_2\mathtt{Gadget}$]
Define $\MCG(t, t', j_1, j_2)$ to be the following set of 2-commodity linear equations representing ``$2x_t = x_{j_1} + x_{j_2}$'':
\begin{align*}
x_t - x_{t'+1} &= 0 \\
x_{t'+2} - x_{j_2} &= 0\\
y_{t'+1} - y_{t'+3} &= 0 \\
y_{t'+4} - y_{t'+2} &= 0 \\
x_{t'+3} - x_{j_1} &= 0\\
x_t - x_{t'+4} &= 0\\
x_{t'+4} - y_{t'+4} - (x_{t'+3} - y_{t'+3}) &= 0\\
x_{t'+1} - y_{t'+1} - (x_{t'+2} - y_{t'+2}) &= 0.
\end{align*}
For convenience we use an extra parameter $t'$ to keep track of new variables that are only used in the gadgets.
\end{definition}

Correctness of this gadget can be easily verified by summing up all the equations in it.

We now present a simplified reduction from exactly solving $\+G_{z,2}$ to exactly solving $\+{MC}_2$ in Algorithm \ref{alg:reduc1}. We use $\*A_i$ to denote the $i$-th row of $\*A$.
\begin{algorithm}[ht]
\DontPrintSemicolon
\SetKwInOut{Input}{Input}\SetKwInOut{Output}{Output}
\SetKwFunction{Gadget}{$\MCG$}

\Input{a $w$-bounded $m \times n$ matrix $\*A \in \+G_{z,2}$ and $\*c \in \RR^n$.}
\Output{a $w$-bit $m' \times n'$ matrix $\*B \in \+{MC}_2$ and $\*d \in \RR^n$.}
\BlankLine
$m' \leftarrow m$, $n' \leftarrow n$\;
$n_g \leftarrow 0$ \tcp*{\# of new variables used only in the $\+{MC}_2\mathtt{Gadget}$s}
\For{$i \leftarrow 1$ \KwTo $m$}{
	\For{$s \in \{+, -\}$}{
		\For{$k \leftarrow 1$ \KwTo $w$}{
			\While{strictly more than 1 entry in $\*A_i$ has sign $s$ and the $k$-th bit being $1$}{
				Let $j_1,j_2$ be the first and second indices of such entries in $\*A_i$\;
				In $\*A_i$, replace $2^k(x_{j_1} + x_{j_2})$ by $2^{k+1}x_{n'+1}$ (thus adding one column to the right of $\*A$)\; \label{alg:ln:replace}
				Add the coefficients of \Gadget{$n'+ n_g +1$, $n'+n_g+1$, $j_1$, $j_2$} to $\*C$\;\label{alg:ln:gadget}
				$n' \leftarrow n' + 1$\;\label{alg:ln:nnv}
				$n_g \leftarrow n_g + 4$\;
				$m' \leftarrow m' + 8$\;
			}
		}
	}
}
$n' \leftarrow 2 \times (n' + n_g)$\;\label{alg:ln:double}
Stack $\*C$ in the bottom of $\*A$ and fill in $0$'s to get $\*B$\;
Add $m' - m$ $0$'s under $\*c$ to get $\*d$\;

\caption{Simplified $\textsc{Reduce} \+G_{z,2}\textsc{To}\+{MC}_2$}\label{alg:reduc1}
\end{algorithm}

Note that in $\+{MC}_2$ we have two input variables sets $X$, $Y$ of the same size. That is why we have to multiply $n'$ by 2 at last in the reduction. But here we will only use variables in $Y$ in the $\MCG$s, and all the other $Y$ variables are unused.
For convenience we arrange the variables in the following way. 
\begin{remark}[Arrangement of variables]\label{rmk:v}
In $\*B$ we put all the variables in $X$ before those in $Y$.
More importantly, we put those $X$ variables that are only used in the gadgets behind all those $x_{n'+1}$ in Line \ref{alg:ln:replace}.
Equivalently it can be viewed as the following process.
We first run Algorithm \ref{alg:reduc1} virtually before Line \ref{alg:ln:double} to get $n'$, which is the number of $X$ variables ignoring those only used in the gadgets.
Then we run it again on the original input, starting with $n_g$ being this value, and in Line \ref{alg:ln:gadget} we use $\+{MC}_2\mathtt{Gadget}(n'+1, n_g, j_1, j_2)$ instead.
\end{remark}

We give an example showing how the reduction works under this arrangement.
\begin{example}\label{exp}
We show how the reduction runs on $3x_1 + 5x_2 + x_3+7x_4 -16x_5 =0$:
\begin{align*}
& 00011 x_1 + 00101 x_2 + 00001 x_3 + 00111x_4 - 10000x_5 = 0 \\
\xrightarrow{x_1 + x_2 - 2x_6 = 0}~& 00010x_1 + 00100x_2 + 00001x_3 + 00111x_4 + 00010x_6 -10000 x_5 =0\\
\xrightarrow{x_3 + x_4 - 2x_7 = 0}~& 00010x_1 + 00100x_2 + 00110x_4 + 00010x_6 + 00010x_7 -10000 x_5 =0\\
\xrightarrow{x_1 + x_4 - 2x_8 = 0}~& 00100x_2 + 00100x_4 + 00010x_6 + 00010x_7 + 00100x_8 -10000 x_5 =0\\
\xrightarrow{x_6 + x_7 - 2x_9 = 0}~& 00100x_2 + 00100x_4 + 00100x_8 + 00100x_9 -10000 x_5 =0\\
\xrightarrow{x_2 + x_4 - 2x_{10} = 0}~& 01000x_{10} + 00100x_8 + 00100x_9 -10000 x_5 =0\\
\xrightarrow{x_8 + x_9 - 2x_{11} = 0}~& 01000x_{10} + 01000x_{11} -10000 x_5 =0\\
\xrightarrow{x_{10} + x_{11} - 2x_{12} = 0}~& 10000x_{12} -10000 x_5 =0.
\end{align*}
We are only eliminating the positive coefficient variables in this example for simplicity.
In the first round we use new variables (and the corresponding gadgets) $x_6$ and $x_7$ to eliminate the first bit, getting the equation $2x_1 + 4 x_2 + 6 x_4 + 2x_6 + 2x_7 = 0$.
These two generated variables are then eliminated in the second round by $x_9$, in addition to $x_8$ for the second bit.
In the third round we use $x_{10}$ and $x_{11}$.
Finally in the fourth round we use $x_{12}$ and get
the equation after the reduction $16 x_{12} - 16 x_5 = 0$, with 
$\+{MC}_2\mathtt{Gadget}$s representing $2x_6 = x_1 + x_2$, $2x_7 = x_3 + x_4$, $2x_8 = x_1 + x_4$, $2x_9 = x_6 + x_7$, $2x_{10} = x_2 + x_4$, $2x_{11} = x_8 + x_9$, and $2x_{12} = x_{10} + x_{11}$.
\end{example}

Correctness of this simplified reduction follows easily by correctness of the gadget, as our transformation preserves the original solution. Moreover, given a solution $\*x^*$ such that $\*B\*x^* = \*d$ where $\*B$ and $\*d$ are obtained from running Algorithm \ref{alg:reduc1} on input $\*A$ and $\*c$, we can easily get the solution to the original equation system $\*A\*x = \*c$ by simply taking the first $n$ elements in $\*x^*$. It is also easy to get $\*d$ from $\*c$ if we can calculate $m'$ in $\TC^0$.

In the remaining of this section we are going to prove that Algorithm \ref{alg:reduc1} can be implemented in $\TC^0$.
For $1 \leq i \leq m$, $1 \leq k \leq w$, we define
\begin{align*}
\Len^+_{i}(\*A) &= \text{ log of sum of positive coefficients in $\*A_i$},\\
\CB^+_{i,k}(\*A) &= \# \text{ of positive coefficient variables in the original $\*A_i$ with the $k$-th bit $1$}, \\
\NNV^+_{i,k}(\*A) &= \# \text{ of $\+{MC}_2\textsc{Gadget}$s used for $\*A_i$ to eliminate the $k$-th bit}\\
&\phantom{= \# }\text{ of positive coefficient variables},
\end{align*}
and similarly $\Len^-_i(\*A)$, $\CB^-_{i,k}(\*A)$, $\NNV^-_{i,k}(\*A)$ for the negative coefficient variables.

\begin{example}
For the above example, we have $\Len^+_i(\*A) = 5$,
and the following values for each $k$.
\begin{table}[h]\caption{Values of $\CB^+_{i,k}(\*A)$ and $\NNV^+_{i,k}(\*A)$ for Example \ref{exp}}
\centering
\begin{tabular}{|c|c|c|c|c|c|}
\hline
$k$ & $1$ & $2$ & $3$ & $4$ & $5$ \\\hline
$\CB^+_{i,k}(\*A)$ & $4$ & $2$ & $2$ & $0$ & $0$ \\\hline
$\NNV^+_{i,k}(\*A)$ & $2$ & $2$ & $2$ & $1$ & $0$ \\\hline
\end{tabular}
\end{table}
\end{example}
We have the following simple but crucial properties for these values.
\begin{claim}\label{clm:obs}
\begin{enumerate}[(i)]
\item $\Len^s_i(\*A) \leq w$ for all $s \in \{+, -\}$, $1 \leq i \leq m$;
\item $\CB^s_{i,k}(\*A) \leq n$ for all $s \in \{+, -\}$, $1 \leq i \leq m$, $1 \leq k \leq w$;
\item For all $s \in \{+, -\}$, $1 \leq i \leq m$,
\[
\NNV^s_{i,k}(\*A) = 
\begin{cases}
2^{-(k+1)} \sum_{k' = 1}^k 2^{k'}\CB^s_{i,k'}(\*A) & \text{ for } 1 \leq k \leq \Len^s_i(\*A)-1, \\
0 & \text{ for } \Len^s_i(\*A) \leq k \leq w,
\end{cases}
\]
thus $\NNV^s_{i,k}(\*A) \leq O(n)$;
\item $m' = m + 8 \sum_{i = 1}^m \sum_{s \in \{+, -\}} \sum_{k = 1}^w \NNV^s_{i,k}(\*A) \leq O(nmw)$;
\item $n' = 2n + 10\sum_{i = 1}^m \sum_{s \in \{+, -\}} \sum_{k = 1}^w \NNV^s_{i,k}(\*A) \leq O(nmw)$.
\end{enumerate}
\end{claim}

\begin{proof}
(i) and (ii) are trivial by definition.
(iv) and (v) are straightforward from Algorithm \ref{alg:reduc1} and (iii).
For (iii), note that in each round for $k$ we will eliminate all the  variables generated in the previous round for $k-1$ by construction (ignoring those variables that are only used in the gadgets), therefore we have $\NNV^s_{i,1}(\*A) = \CB^s_{i,1}(\*A)/2$ and
$\NNV^s_{i,k}(\*A) = (\CB^s_{i,k}(\*A) + \NNV^s_{i, k-1}(\*A))/2$ for $2\leq k \leq \Len^s_i(\*A) -1$.
Here we rely on the property that the sum of positive (and negative) coefficients in each row is a power of $2$ to ensure that $\NNV^s_{i,k}(\*A)$ as calculated in this way are always integers.
Then by induction we get the formula.  
\end{proof}

Note that in (iii) $2^{-(k+1)}$ and $2^{k'}$ are just right and left shifts, which can be easily implemented in the circuit model. 
Combining Claim \ref{clm:obs} and Fact \ref{fat:add}, we can see that all of these values can be computed in $\TC^0$ for all $i, k, s$, i.e. the depths of the $\TC$ circuits are absolute constants independent of $i$, $k$, and $s$.
\begin{lemma}[Informal]\label{lem:aux}
\begin{itemize}
\item $\Len^s_i$, $\CB^s_{i,k}$, and $\NNV^s_{i,k}$ have $\TC^0$ circuits of size $\poly(n,w)$ for all $i, k, s$;
\item $n', m'$ has $\TC^0$ circuits of size $\poly(n,m,w)$.
\end{itemize}
\end{lemma}


Now we can prove that the reduction from $\+G_{z,2}$ to $\+{MC}_2$ can be done in $\TC^0$. We represent the input matrix $\*A$ explicitly by giving all its entries in sign-magnitude form, and the output matrix $\*B$ implicitly by outputting the entry of $\*B$ in row $i$ column $j$ with $i$,$j$ as additional inputs.
\begin{theorem}
There is a $\TC^0$ circuit family $\{C_{n,m,w}\}_{n,m,w \in \mathbb{N}}$ of size $\poly(n,m,w)$ with $(O(nmw) + O(\log nmw))$-bit input and $O(w)$-bit output such that:
\begin{itemize}
\item the first $O(nmw)$ bits of input represent a $w$-bounded $m \times n$ matrix $\*A \in \+G_{z,2}$ in sign-magnitude form, the next $O(\log nmw)$ bits of input represent an index $i$, and the last $O(\log nmw)$ bits represent another index $j$;
\item for $i \leq m'$ and $j \leq n'$ where $m', n' \leq O(nmw)$ is calculated as in Algorithm \ref{alg:reduc1} on $\*A$,  the output represents the entry of $\*B$ (in sign-magnitude form) in row $i$ column $j$ as calculated by Algorithm \ref{alg:reduc1} on $\*A$; otherwise it represents $0$.
\end{itemize} 
\end{theorem}

\begin{proof}
Consider the $i$-th row of $\*B$, we need to know the equation it corresponds to. Because we put those equations of $\MCG$s behind the modified equations of $\*A$, we have two cases:
\begin{itemize}
\item $i \leq m$: in this case, the equation is $\*A_i$ after the transformation, which has the form $C (x_{j_+} - x_{j_-}) = \*c_i$ with $C = \underbrace{100\cdots 0}_{\Len^+_i(\*A)\text{ bits}}$ in binary since our transformation does not change the sum of positive (and negative) coefficients in $\*A$. What remains is to calculate the indices $j_+$ and $j_-$ in $\TC^0$.

Let $\SNNV^s_i(\*A) = \sum_{k=1}^w \NNV^s_{i,k}$ for all $s \in \{+,-\}$ and $1 \leq i \leq m$. For each $s \in \{+, -\}$, there are two cases:
	\begin{itemize}
	\item No new variable is added, i.e. $\SNNV^s_i(\*A) = 0$: there is only one entry in the original $\*A_i$ with the corresponding sign thus $j_s$ is the index of this entry. We search for this entry to get its index, which can be implemented in $\AC^0$ because there are $n$ possibilities.
	\item Some new variables are added: then $j_s$ is the index of the last added new variable, ignoring those only in $\MCG$s. It can be calculated by
\begin{align*}
j_s = \begin{cases}
n + \SNNV^+_i(\*A) + \sum_{i'=1}^{i-1}\sum_{s'\in\{+,-\}}\SNNV^{s'}_{i'}(\*A) &\text{ if } s = +, \\
n + \sum_{i'=1}^{i}\sum_{s'\in\{+,-\}}\SNNV^{s'}_{i'}(\*A) &\text{ if } s = -.
\end{cases}
\end{align*}
	\end{itemize}

$\SNNV^s_i \in \TC^0$ by Fact \ref{fat:add} and Lemma \ref{lem:aux} so both $j_+, j_-$ can be calculated in $\TC^0$ by Fact \ref{fat:add}. Therefore we can calculate the entries in row $i$ in $\TC^0$ for $i \leq m$.

\item $i > m$: in this case, this equation is in an $\MCG(t, t', j_1, j_2)$ gadget for some $t, t', j_1, j_2$. We need to calculate $t, t', j_1$, and $j_2$ in $\TC^0$ then it is easy to recover the equation from the offset in the gadget.

We can first calculate in $\AC^0$ the gadget's index $ind = \lfloor \frac{i-m-1}{8}\rfloor+1$ as $\lfloor \frac{\cdot}{8}\rfloor$ is just ignoring the three least significant bits.
By Algorithm \ref{alg:reduc1} and Remark \ref{rmk:v}, we know $t = n + ind$ thus it is $\AC^0$-computable,
and we can also calculate $t' = n + \sum_{i'=1}^m\sum_{s' \in \{+,-\}} \SNNV^{s'}_{i'}(\*A) + 4 (ind - 1)$ in $\TC^0$ by Fact \ref{fat:add}.

Then we can calculate the number $i'$, $k'$, sign $s' \in \{+,-\}$, and number $\ell'$ such that this gadget is the $\ell'$-th gadget we created to eliminate the $k'$-th bit of $\*A_{i'}$ for variables with sign $s'$. More specifically, we want to find the minimum $i'$, $k'$, $s'$ (where we treat `$+$' $<$ `$-$') such that 
\begin{align*}
\PreS^+_{i',k'}(\*A)
 < ind \leq \PreS^+_{i',k'}(\*A) + \NNV^+_{i',k'} & \text{ if } s' = +,\\
\PreS^-_{i',k'}(\*A) < ind \leq \PreS^-_{i',k'}(\*A) + \NNV^-_{i',k'} & \text{ if } s' = -,
\end{align*}
where
\begin{align*}
\PreS^+_{i',k'}(\*A) &= \sum_{i''=1}^{i'-1}\sum_{s''\in\{+,-\}}\SNNV^{s''}_{i''}(\*A)
+\sum_{k''=1}^{k'-1}\sum_{s''\in\{+,-\}}\NNV^{s''}_{i',k''}(\*A),\\
\PreS^-_{i',k'}(\*A) &= \PreS^+_{i',k'}(\*A) + \NNV^+_{i',k'}.
\end{align*}

There are $O(mw)$ possible choices for $i'$, $k'$, $s'$ and each condition is a prefix sum of $\NNV^s_{i,k}$'s, therefore
this can be done in $\TC^0$ by a parallel comparison of $ind$ to the prefix sums of $\NNV^s_{i,k}$'s. After getting $i'$, $k'$, and $s'$, we can get $\ell' = ind - \PreS^{s'}_{i',k'}$ by Fact \ref{fat:add}.

What remains is to calculate $j_1$ and $j_2$ from $i'$, $k'$, $s'$, and $\ell'$. Note that when eliminating the $k'$-th bit of $\*A_{i'}$ for variables with sign $s'$, we first eliminate in pairs those variables in the original $\*A_{i'}$ that have $1$ in the $k'$-th bit before the reduction (we call them \emph{original pairs}), then eliminate in pairs those generated in the previous round for $i'$, $k'-1$, and $s'$. The number of original pairs is given by $p = \lfloor \CB^{s'}_{i',k'}(\*A) / 2\rfloor$, computable in $\TC^0$. There are two cases:
	\begin{itemize}
	\item $\ell' \leq p$: $j_1$ and $j_2$ is the indices of variables in the $\ell'$-th original pairs, which are the indices of the $(2\ell'-1)$-th and $2\ell'$-th variables in $\*A_{i'}$ that have $1$ in the $k'$-th bit. Similarly as above, this can be done by a simple parallel comparison to prefix sums of the $k'$-th bits of variables in $\*A_{i'}$.
	\item $\ell' > p$: then $j_1$ and $j_2$ are the $(2(\ell'-p)-1)$-th and $2(\ell'-p)$-th new variables generated in the previous round, therefore $j_1 = \PreS^{s'}_{i', k'-1}(\*A) + 2(\ell'-p)-1$ and $j_2 = j_1 + 1$.
	\end{itemize}
However the second case above only works for even $\CB^{s'}_{i',k'}(\*A)$. When it is odd, $\ell' = p+1$ corresponds to a pair with the last original variable in this round and the first generated variables in the previous round, thus $j_1$ can be calculated as in the first case and $j_2 = \PreS^{s'}_{i', k'-1}(\*A) + 1$. Then for $\ell' > p+1$, we have $j_1 = \PreS^{s'}_{i', k'-1}(\*A) + 2(\ell'-p)-2$ and $j_2 = j_1 + 1$.
\end{itemize}

In conclusion, given $i$ as input, we can first check if $i \leq m'$ in $\TC^0$ by Lemma \ref{lem:aux} and if so compute coefficients of $\*B_i$ in $\TC^0$, therefore with $j$ as input, we check if $j \leq n'$ in $\TC^0$ and if so compute the entry of $\*B$ in row $i$ column $j$ in $\TC^0$.
\end{proof}

\section{Genearlization to approximate solvers, and more restrictive classes}
Now we are going to generalize the above results of reductions for exact solvers into those for approximate solvers, thus proving Theorem \ref{thm:main}. First we need the following result showing the power of $\TC^0$.
\begin{fact}\label{fat:tc0}
Division and iterated multiplication are in $\mathsf{DLOGTIME}$-uniform $\TC^0$ \cite{HesseAB02}. Moreover, we can approximate in $\TC^0$ functions represented by sufficiently nice power series, such as $\log$, $\exp$, and $x^{1/k}$ \cite{ReifT92,MacielT99,HesseAB02}.
\end{fact}

\begin{proof}[Proof sketch of Theorem \ref{thm:main}]
Based on our proofs on the simplified reductions, we are going to prove that the original reductions from $\+G$ to $\+{MC}_2$ in \cite{KyngZ17} can be computed in $\TC^0$.
In the context of approximate solvers, the error $\e$ that solvers are required to achieve is also part of the instance.
By Fact \ref{fat:tc0}, all the errors in the reduced instances, as defined in the reductions in Kyng and Zhang's full paper \cite{KyngZ17full}, can be computed in $\TC^0$ given the size, condition number, and bit-complexity of the original matrix as parameters.
\begin{itemize}
\item \textbf{From $\+G$ to $\+G_z$}: this reduction remains the same, but now knowing the accuracy we can recover $\*x'$ from $\*x$ using fix-point arithmetic in $\TC^0$.
\item \textbf{From $\+G_z$ to $\+G_{z,2}$}: the only difference (besides the calculation of accuracy) is that in the original reduction in \cite[Section~7.2]{KyngZ17full}, in the last row of the reduced matrix $\*A'$ the last two entries are set to be $w$ and $-w$ for some value $w$, instead of $1$ and $-1$ as in our simplified version. However $w$ is computable in $\TC^0$ by Fact \ref{fat:tc0} so the reduction is still computable in $\TC^0$.
\item \textbf{From $\+G_{z,2}$ to $\+{MC}_2$}: for approximate solvers we will have to use the original $\+{MC}_2\mathtt{Gadget}$ in \cite[Section~4]{KyngZ17full} consisting of ten $2$-commodity equations instead of eight and with additional variables.
So we need to modify the corresponding numbers in our calculation, and in particular the gedget's index becomes $ind = \lfloor \frac{i -m - 1}{10} \rfloor +1$, which is still $\TC^0$-computable.
Besides, the equations in the $\+{MC}_2\mathtt{Gadget}$ for eliminating the $k$-th bit of variables with sign $s$ in $\*A_i$ will be multiplied by a factor $-s \cdot 2^k \cdot w_i$, where $w_i = \sqrt{10 \sum_{s' \in \{+,-\}} \SNNV^s_i(\*A)}$, as specified in Algorithm 1 of \cite{KyngZ17full}.
These factors can be calculated with desired accuracy in $\TC^0$ by Fact \ref{fat:tc0}.
Therefore the reduction is still computable in $\TC^0$.
\end{itemize}
Additionally in the second and third reduction, when recovering $\*x'$ from $\*x$ we need to check if the original matrix $\*A$ and vector $\*b$ satisfy $\*A^\top \*b = \*0$ and simply return $\*0$ if so. This can be done in $\TC^0$ by Fact \ref{fat:add}.
In conclusion, we can reduce the problem of approximately solving equations in $\+G$ to approximately solving equations in $\+{MC}_2$ in $\TC^0$.
\end{proof}

In \cite{KyngZ17} they also considered some more restrictive subclasses of $\+{MC}_2$.
Intuitively, the set of \emph{strict $2$-commodity matrices} $\+{MC}^{>0}_2 \subset \+{MC}_2$ is the class of 2-commodity equations $\*A$ such that for every pair $i, j$, equation $x_i - x_j = 0 \in \*A$ iff equation $y_i - y_j = 0 \in \*A$ iff equation $x_i - y_i - (x_j - y_j) = 0 \in \*A$.
The set of strict $2$-commodity matrices with integer entries is denoted by $\+{MC}^{>0}_{2,\ZZ}$.
They showed that reductions from approximately solving $\+{MC}_2$ to approximately solving $\+{MC}^{>0}_2$, and from $\+{MC}^{>0}_2$ to $\+{MC}^{>0}_{2,\ZZ}$.
We are going to show that these reductions can be computed in $\TC^0$.
\paragraph*{From $\+{MC}_2$ to $\+{MC}^{>0}_2$}
\begin{proof}[Proof sketch]
The reduction, as defined in \cite[Section~5.1]{KyngZ17full}, runs by checking for each pair $i,j$ that are involved in some equation in $\*A$ if any of the three types of equations is missing and add it if so.
The added equations will be multiplied by a factor that is computable in $\TC^0$. Obviously the resulting equation systems is in $\+{MC}^{>0}_2$. It is easy to see that the number of added equations for each pair $i,j$ can be computed in $\AC^0$, thus all the prefix sums of these numbers can be calculated in $\TC^0$ simultaneously, and so we can determine the equations in the reduced instance in $\TC^0$.
\end{proof}
\paragraph*{From $\+{MC}^{>0}_2$ to $\+{MC}^{>0}_{2,\ZZ}$}
\begin{proof}[Proof sketch]
In \cite[Section~6]{KyngZ17full} it is done by scaling up all the numbers in the matrix and the vector by a factor of $2^k$, where $k$ is computable in $\TC^0$ by Fact \ref{fat:tc0}, then take the ceiling function on entries of the matrix to convert them into integer entries, which also can be done in $\TC^0$.
\end{proof}

\bibliography{tc0}

\end{document}